\pdfoutput=1
\documentclass[11pt]{article}

\usepackage[letterpaper,margin=1in]{geometry}
\usepackage{bm, amsmath, amssymb, amsthm, amsfonts}
\usepackage{hyperref}
\usepackage[capitalize,nameinlink]{cleveref} 
\usepackage{cite}
\usepackage{framed}
\usepackage[framemethod=tikz]{mdframed}
\usepackage{appendix}
\usepackage{graphicx}
\usepackage{color}
\usepackage{varwidth}
\usepackage{wrapfig}
\usepackage{algorithm2e}
\usepackage[noend]{algpseudocode}
\usepackage[textsize=tiny]{todonotes}

\usepackage{enumitem}
\setitemize{noitemsep,topsep=3pt,parsep=3pt,partopsep=3pt}

\usepackage{xspace}
\usepackage{xifthen}

\definecolor{darkgreen}{rgb}{0,0.5,0}
\definecolor{darkblue}{rgb}{0,0,0.8}
\usepackage{hyperref}

\hypersetup{
	colorlinks=true,                  
	linkcolor=darkblue,
	citecolor=darkgreen
}

\newtheorem{theorem}{Theorem}[section]
\newtheorem{lemma}[theorem]{Lemma}

\newcommand{\ignore}[1]{}

\algnewcommand\algorithmicswitch{\textbf{switch}}
\algnewcommand\algorithmiccase{\textbf{case}}

\algdef{SE}[SWITCH]{Switch}{EndSwitch}[1]{\algorithmicswitch\ #1\ \algorithmicdo}{\algorithmicend\ \algorithmicswitch}%
\algdef{SE}[CASE]{Case}{EndCase}[1]{\algorithmiccase\ #1}{\algorithmicend\ \algorithmiccase}%
\algtext*{EndSwitch}%
\algtext*{EndCase}%

\newcommand{\LOCAL}{\ensuremath{\mathsf{LOCAL}}\xspace}
\newcommand{\SLOCAL}{\ensuremath{\mathsf{SLOCAL}}\xspace}

\newcommand{\eps}{\varepsilon}
\newcommand{\poly}{\operatorname{\text{{\rm poly}}}}

\DeclareMathOperator{\polylog}{\poly\log}

\newcommand{\complexityclass}[2][]{\ensuremath{\mathsf{#2}\ifthenelse{\isempty{#1}}{}{(#1)}}}

\newcommand{\pslocal}{\complexityclass{P}\text{-}\complexityclass{SLOCAL}\xspace}

\newcommand{\hide}[1]{}

\newcommand{\FullOrShort}{short}

\ifthenelse{\equal{\FullOrShort}{full}}{
	
  \newcommand{\fullOnly}[1]{#1}
  \newcommand{\shortOnly}[1]{}

}{

  \newcommand{\shortOnly}[1]{#1}
  \newcommand{\fullOnly}[1]{}

}
  
\begin{document}

\title{$\pslocal$-Completeness of Maximum Independent Set Approximation}

\author{
Yannic Maus ---
Technion, Israel ---
\small\tt yannic.maus@campus.technion.ac.il
}

\date{}

\maketitle

\thispagestyle{empty}

\begin{abstract}
 We prove that the maximum independent set approximation problem with polylogarithmic approximation factor is \pslocal-complete.
\end{abstract}

\section{Introduction}
To formally present our results we recall the model of computation. 
\paragraph{The \LOCAL Model of distributed computing~\cite{linial92}.} A graph is abstracted as an $n$-node network $G=(V, E)$ with maximum degree $\Delta$. Communications happen in synchronous rounds. Per round, each node can send one (unbounded size) message to each of its neighbors. At the end, each node should know its own part of the output, e.g., whether it belongs to an independent set or its own color in a vertex coloring. Typically, an algorithm
in the \LOCAL model is considered \emph{fast} or \emph{efficient} if its time complexity is
polylogarithmic in the number of nodes $n$.

The maximal independent set  problem (MIS)\footnote{An \emph{MIS} of a graph $G=(V,E)$ is an inclusion maximal independent set $M\subseteq V$. A \emph{maximum independent set (MaxIS)} is an independent set with maximum cardinality; this cardinality is denoted by the \emph{independence number} $\alpha(G)$. A $\lambda$-approximation of a maximum independent set is an independent set of size at least $\alpha(G)/\lambda$~.}  and the $(\Delta+1)$-vertex coloring problem have fast randomized algorithms \cite{luby86} and exponentially slower deterministic algorithms \cite{awerbuch89}. 
The question whether the MIS problem has a polylogarithmic time deterministic algorithm dates back to Linial's seminal paper \cite{linial92} and is still considered to be one of the most important open questions of the area \cite{MausThesis,barenboimelkin_book}. 
With the objective to shed some light on \emph{the role of randomness} in solving problems like MIS or vertex coloring efficiently Ghaffari, Kuhn and Maus \cite{SLOCAL17} defined a sequential variant of the \LOCAL model, the \SLOCAL model. In an \SLOCAL algorithm with complexity (or locality) $r$ the nodes of the network graph are processed in an arbitrary order. When a node $v$ is processed it can see the current state of all nodes in its $r$-hop neighborhood (including all topological information of this neighborhood) and its output can be an arbitrary function of this neighborhood. Additionally, it can store information that can be read by later nodes as part of $v$' state. The maximal independent set problem admits an \SLOCAL algorithm with locality $r=1$ by iterating through the nodes in an arbitrary order and joining the independent set if none of the already processed neighbors is already contained in the set. An \SLOCAL algorithm with constant locality is not known for the maximum independent set problem.
Then, \cite{SLOCAL17} defines the class \pslocal, i.e., the class of problems that can be solved with polylogarithmic complexity in the \SLOCAL model. Besides MIS and vertex-coloring the class contains many other classic and related problems.  
In particular, it contains all problems that can be solved efficiently by randomized algorithms in the \LOCAL model as long as a solution of the problem can be verified efficiently \cite{newHypergraphMatching}. 
Most important for our results is the notion of  \pslocal-completeness: A problem is \emph{\pslocal-complete} if it is contained in the class \pslocal and if it is \emph{\pslocal-hard}, that is, any other problem in the class can be efficiently reduced to it via a local reduction.\footnote{For the formal definitions of reductions, the model and the complexity class we refer to \cite{SLOCAL17}. For this work it is sufficient to think of a reduction from problem $\mathcal{B}$ to problem $\mathcal{A}$ as a \LOCAL algorithm that uses an algorithm for problem $\mathcal{A}$ to solve problem $\mathcal{B}$ while only incurring a polylogartihmic overhead.} If any \pslocal-complete problem can be solved efficiently by a deterministic algorithm in the \LOCAL model all problems in the class \pslocal can be solved efficiently by deterministic algorithms; this includes the MIS and vertex coloring problem.



Among few others the following  are \pslocal-\-complete problems: $(\polylog n,\polylog n)$-network decompositions, (weak) local splittings \cite{SLOCAL17}, approximations of dominating set and distributed set cover \cite{newHypergraphMatching}. However, it would be most interesting if the MIS or the vertex coloring problem were shown to be complete \cite{MausThesis}.
 We cannot show either one but we show that computing sufficiently good approximations for the \underline{maximum} independent set problem is complete.
\begin{theorem}\label{thm:mainMaxISComplete}
Polylogarithmic maximum independent set approximation is \pslocal-complete.
\end{theorem}

The proof of \Cref{thm:mainMaxISComplete}  is through a reduction from the conflict-free multicoloring problem, one of the first problems known to be \pslocal-complete \cite{SLOCAL17}. The objective of the conflict-free $k$-coloring problem is to compute a vertex coloring $f:V\rightarrow \{1,\ldots,k\}$ of a hypergraph $H=(V,E)$  such that for each edge $e\in E$ there exists a node $v\in e$ with a unique color, that is, there is no $u\neq v$ with $u\in e$ and $f(u)=f(v)$. 
We call an edge with this property \emph{happy} in coloring $f$---in proofs we consider colorings in which only some edges are happy. 
In the conflict-free multicoloring problem each node is allowed to have more than one color and all other requirements are the same. For a given constant $0<\eps\leq 1$ we call a hypergraph $H=(V,E)$ \emph{almost uniform} if there is an arbitrary $k$ such that for all edges $e\in E$ we have $k\leq |e|\leq (1+\eps)k$~.
\begin{theorem}[\cite{SLOCAL17}]
\label{thm:conflictComplete}
Conflict-free multicoloring with $\polylog n$
colors in almost uniform hypergraphs with $\poly n$ hyperedges is
\pslocal-complete.
\end{theorem}
The unpublished work \cite{2018arXivConflict} uses the MaxIS to solve conflict-free coloring on so called \emph{interval hypergraphs}. To show the \pslocal-hardness of MaxIS approximations we adapt their techniques and reduce the conflict-free multicoloring problem to the problem of computing MaxIS approximations.


\section{Completeness of Maximum Independent Set Approximation}
We define the \emph{conflict graph} $G_k$ of conflictfree $k$ coloring a hypergraph $H$:
The vertex set $V(G_k)$ consists of all triples $(e,v,c)$, $e\in E(H), v\in e, 1\leq c \leq k$. 
The edge set $E(G_k)$ is
\begin{align*}
E_{vertex} & =\left\{\{(e,v,c),(g,v,d)\} \mid v\in V(H), 1\leq c\neq d \leq k\right\}~\cup \\
E_{edge} & =\left\{\{(e,v,c),(e,u,d)\} \mid e\in E(H),~ u,v\in e, ~1\leq c, d \leq k\right\}~\cup \\
E_{color} & =\big\{\{(e,v,c),(g,u,c)\} \mid e,g\in E(H),~1\leq c\leq k,   \\
& \hspace{3.5cm}\{u,v\}\subseteq e \text{ or } \{u,v\}\subseteq g,~ \big\}~.
\end{align*} 

We explain how any independent set of $G_k$ naturally corresponds to a partial vertex coloring of the hypergraph $H$ and vice versa: Given an independent set $\mathcal{I}$ of $G_k$ define the partial vertex coloring:
\begin{align} \label{eqn:colorAssignment}
f & : V\rightarrow \{1,\ldots,k\}\cup \{\bot\}, \nonumber \\
f_{\mathcal{I}}(v) & = 
\left\{\begin{array}{lr}
c, & \exists e \text{ such that } (e,v,c)\in \mathcal{I}, \\
\bot, & otherwise~.
\end{array}\right.
\end{align} 
Vice versa, given a conflictfree coloring $f$ of $H$ with $k$ colors we construct an independent set of $G_k$ as follows: For each edge $e$  we add a single node $(e,v,c)$---breaking ties arbitrarily---to the independent set $\mathcal{I}_f$ where $v\in e$, $f(v)=c$ and there is no $u\neq v$ with $u\in e$ and $f(u)=c$.
More formally, we prove the following statement about the correspondence. 
\begin{lemma}\label{lem:technical}
Let $k>0$, $H$ be a hypergraph and $G_k$ the conflict graph of conflictfree $k$-coloring $H$. 
\begin{enumerate}[label=\alph*)]
\item Any conflict-free $k$-coloring $f$ of $H$ induces a maximum independent set $\mathcal{I}_f$ of the conflict graph $G_k$. The size of this maximum independent set is $m=|E(H)|$.
\label{lem:ColoringToIndependentSet}
\item For any independent set $\mathcal{I}\subseteq V(G_k)$ the induced coloring $f_{\mathcal{I}}$ is well defined and at least $|\mathcal{I}|$ edges of $H$ are happy in $f_{\mathcal{I}}$. 
\label{lem:IndependentSetToColoring}
\end{enumerate}
\end{lemma}
\begin{proof}
\noindent \textit{Proof of \ref{lem:ColoringToIndependentSet}: }
We first show by a proof of contradiction that $\mathcal{I}_f$ is an independent set. So, let $\mathfrak{v}_1=(e,v,c)$ and $\mathfrak{v}_2=(g,u,d)$ be two nodes in $\mathcal{I}_f$ and assume that they are connected by an edge $h\in E(G_k)$. We perform a case distinction depending on whether $h\in E_{vertex}$, $h\in E_{edge}$ or $h\in E_{color}$ and deduce a contradiction in all three cases. 

$h\in E_{vertex}$: This means that $v=u$ and also that $c\neq d$. However, this would mean that edge $e$ added node $\mathfrak{v}_1$ to $\mathcal{I}_f$ because $v$ is the unique node with color $c$ in coloring $f$ and edge $g$ added $\mathfrak{v}_2$ to $\mathcal{I}_f$ because $v$ is the unique node with color $d$ in coloring $f$, a contradiction to $v$ having only one color in $f$. 

$h\in E_{edge}$: This means that $e=g$, i.e., edge $e$ added $\mathfrak{v}_1$ and $\mathfrak{v}_1$  to the independent set $\mathcal{I}_f$. However, each edge only adds one node to the independent set $\mathcal{I}_f$, a contradiction.

$h\in E_{color}$: This means that $c=d$, i.e., the nodes in the independent set are $\mathfrak{v}_1=(e,v,c)$ and $\mathfrak{v}_2=(g,u,c)$, and  $\mathfrak{v}_1$ is added to $\mathcal{I}_f$ by edge $e$ and $\mathfrak{v}_2$ is added to $\mathcal{I}_f$ by edge $g$. Furthermore, $\{v,u\}\subseteq e$ or $(\{v,u\}\subseteq g$ holds. Without loss of generality assume that $\{v,u\}\subseteq e$ holds. Edge $g$ added $\mathfrak{v}_2$ to the independent set $\mathcal{I}_f$ because $u$ is the only node among the nodes in edge $g$ with color $c$. Edge $e$ added $\mathfrak{v}_1$ to the independent set $\mathcal{I}_f$ because $v$ is the only node among the nodes in edge $e$ with color $c$ in $f$, a contradiction as $u\in e$ and $u$ also has color $c$ in $f$.

\textit{$\mathcal{I}_f$ is a maximum independent set:} As $f$ is a conflictfree coloring each edge $e$ adds one node of the form $(e,\star, \star)$ to the independent set. A node cannot be added to the independent by two distinct edges. Thus $|\mathcal{I}_f|\geq m=|E(H)|$. Further, any independent set of $G_k$ has at most $m$ nodes as any node of the form $(e,\star, \star)$ in the independent set forbids any other node of that form in the independent set due to the definition of $E_{edge}$. 

\medskip 

\noindent\textit{Proof of \ref{lem:IndependentSetToColoring}:}
We first show that the color assignment in \Cref{eqn:colorAssignment} is well defined, that is, for each node $u$ there is at most one $c$ such that $(\star,u,c)\in \mathcal{I}$. For contradiction, assume that there are two nodes $\mathfrak{v}_1$ of type $(\star,v,c)$ and $\mathfrak{v}_2$ of type $(\star,v,d)$ with $c\neq d$ in the independent set $\mathcal{I}$. By $E_{vertex}$ there is an edge in $G_k$ between $\mathfrak{v}_1$ and $\mathfrak{v}_2$, a contradiction to both nodes being in $\mathcal{I}$.

\textit{At least $|\mathcal{I}|$ edges are happy in $f_{\mathcal{I}}$: }  Let $\mathfrak{v}=(e,v,c)\in \mathcal{I}$ be a node of the independent set. We prove that edge $e$ is happy: Node $v$ is colored with color $c$  in $f_{\mathcal{I}}$; assume that there is a further node $u\in e, u\neq v$ that is colored with color $c$. By the definition of $f_{\mathcal{I}}$ this means that there is some node $\mathfrak{v}'=(g,u,c)$ in the independent set. As $\{u,v\}\subseteq e$ the nodes $\mathfrak{v}$ and $\mathfrak{v}'$ are connected in $G_k$ due to the definition of $E_{color}$, a contradiction. 
Thus for each node $(e,\star,\star)\in\mathcal{I}$ the edge $e$ is happy. 
Further, for each happy edge $e$ there is at most one node $(e,\star,\star)\in \mathcal{I}$ due to the definition of $E_{edge}$. Thus the number of happy edges equals the size of the independent set $\mathcal{I}$.
\end{proof}

The conflict graph $G_k$ can be efficiently simulated in $H$ in the \LOCAL model. The iterative computation of maximum independent set approximations in suitably chosen conflict graphs combined with  \Cref{lem:technical} are sufficient to prove the hardness in \Cref{thm:mainMaxISComplete}.
\begin{proof}[Proof of \Cref{thm:mainMaxISComplete}]
The containment was proven in \cite[Theorem 7.1]{SLOCAL17}. 

To prove the hardness we reduce from the \pslocal-complete conflictfree multicoloring problem (\Cref{thm:conflictComplete}). Assume that we can compute $\lambda$-approximations for MaxIS and let $H$ be a hypergraph with $m=\poly n$ edges as in \Cref{thm:conflictComplete}. First note, that the graphs used for the hardness in the proof of \Cref{thm:conflictComplete} all admit a conflictfree $k$-coloring where each node only has a \underline{single} color and $k=\polylog n$; fix this $k$ and let $\rho=\lambda\cdot \ln m +1$.
In the reduction we use phases $1,\ldots,\rho$ and in each phase we color some of the vertices in $V$ using a distinct palette of size $k$ for each phase; after each phase we remove all happy edges from the graph. We continue with describing the phases in more detail. For $i=1,\ldots,\rho$ let $H_i=(V,E_i)$ denote the hypergraph in phase $i$ where $H_1=H=(V,E)$ is the original graph.
In  phase $i$ we use the hypergraph $H_i=(V,E_i)$ to build the conflict graph $G^i_k$. $G^i_k$ has polynomially many nodes and edges and can be simulated locally. Then we compute an independent set $\mathcal{I}^i$ of $G^i_k$ that is a $\lambda$-approximation for MaxIS. Each node $v\in V(H)$ that has some $(v,\star,c)\in \mathcal{I}^i$ colors itself with color $c$ (using a distinct palette of size $k$ for each phase); all other nodes remain uncolored in this phase. The algorithm continues with the next phase.

We continue with analyzing the number of edges of the hypergraph that become happy, that is, are removed per phase. $G^{i}_k$ has an independent set of size $|E_i|$ because of \Cref{lem:technical}, \ref{lem:ColoringToIndependentSet} and as $H$ and also $H_i\subseteq H$ admit a conflictfree $k$-coloring . As $\mathcal{I}^i$ is a $\lambda$-approximation for MaxIS we obtain  $|\mathcal{I}^i|\geq \frac{|E_i|}{\lambda}$.
Due to \Cref{lem:technical}, \ref{lem:IndependentSetToColoring} there is a distinct happy edge in $E_i$ for each node in $\mathcal{I}^i$. Thus $|E_{i+1}|\leq |E_i|-|\mathcal{I}^i|\leq (1-1/\lambda)|E_i|$. Thus after $\rho$ phases  we have $\|E_{\rho+1}|\leq (1-1/\lambda)^{\rho}|E|\leq e^{\frac{\rho}{\lambda}}m<1$, i.e., all edges of the initial hypergraph $H$ are happy and removed. Thus the obtained multicoloring is conflictfree and the total number of colors is $k\cdot \rho=\polylog n$.
\end{proof}

It remains open whether the MIS or the $(\Delta+1)$ vertex coloring problem are \pslocal-complete.

\bibliographystyle{alpha}
\bibliography{references}

\appendix

\include{appendix}

\end{document}